\documentclass[11pt,twoside]{article} 
\usepackage{fancyhdr}
\usepackage[utf8]{inputenc}
\usepackage{latexsym}
\usepackage[english]{babel}
\usepackage[mathscr]{eucal}
\usepackage{amsmath}
\usepackage{subcaption}
\usepackage{amssymb}
\usepackage{braket}
\usepackage{titlesec} \titleformat{\section}{\centering\normalsize\normalfont\scshape}{\thesection.}{0.5em}{}
\usepackage{bbm}
\usepackage[T1]{fontenc}
\usepackage{graphicx}
\usepackage{enumerate}
\usepackage{xcolor}
\usepackage{authblk}
\usepackage{pgfplots} \usetikzlibrary{plotmarks} 
\usepackage{mathrsfs}
\usepackage{amsthm}
\usepackage{amsfonts}
\usepackage[margin=1.9cm,bottom=1.8cm,top=1.3cm,includehead,headsep=0.5cm]{geometry} 
\usepackage{amssymb}
\usepackage{hyperref} 

\newcommand{\R}{\mathbb{R}}  
\ifdefined\C\renewcommand{\C}{\mathbb{C}}\else\newcommand{\C}{\mathbb{C}}\fi 
\newcommand{\N}{\mathbb{N}}  
\renewcommand{\o}[1]{\scriptstyle\mathcal{O}\textstyle\left(#1\right)} 
\renewcommand{\O}[1]{\mathcal{O}\left(#1\right)} 
\newcommand*{\defeq}{\mathrel{\vcenter{\baselineskip0.5ex \lineskiplimit0pt\hbox{\scriptsize.}\hbox{\scriptsize.}}}=}
                     
\providecommand{\abs}[1]{\lvert#1\lvert} 

\DeclareMathOperator{\E}{\mathbf{E}}
\DeclareMathOperator{\Tr}{Tr}

\makeatletter
\providecommand*{\diff}%
        {\@ifnextchar^{\DIfF}{\DIfF^{}}}
\def\DIfF^#1{%
        \mathop{\mathrm{\mathstrut d}}%
                \nolimits^{#1}\gobblespace
}
\def\gobblespace{%
        \futurelet\diffarg\opspace}
\def\opspace{\let\DiffSpace\! \ifx\diffarg(\let\DiffSpace\relax\else\ifx\diffarg\let\DiffSpace\relax\else\ifx\diffarg\{\let\DiffSpace\relax\fi\fi\fi\DiffSpace}   

\newtheorem{theorem}{Theorem}
\newtheorem{lemma}[theorem]{Lemma}
\newtheorem{prop}[theorem]{Proposition}
\theoremstyle{remark}
\newtheorem{remark}[theorem]{Remark}

\makeatletter \newtheorem*{rep@theorem}{\rep@title} \newcommand{\newreptheorem}[2]{ \newenvironment{rep#1}[1]{ \def\rep@title{\bf #2 \ref{##1}'} \begin{rep@theorem} } {\end{rep@theorem} } } \makeatother \newreptheorem{theorem}{Theorem}

\title{\uppercase{\bfseries\large Phase Transition in the Density of States of Quantum Spin Glasses}}
\newcommand\shorttitle{Density of States for Spin Glasses}
\newcommand\authors{László Erdős and Dominik Schröder}

\author[1]{László Erdős\thanks{Partially supported by ERC Advanced Grant No. 338804, Email: lerdos@ist.ac.at}}
\author[2]{Dominik Schröder\thanks{Email: schroeder.dominik@gmail.com}}
\affil[1]{IST Austria, Am Campus 1, Klosterneuburg A-3400}
\affil[2]{Mathematisches Institut, Ludwig-Maximilians-Universität, Theresienstraße 39, D-80333 München}

\date{\today}

\begin{document}

\clearpage\maketitle\thispagestyle{empty}\renewcommand{\abstractname}{\vspace{-1.8cm}}
\begin{abstract}\noindent {\scshape Abstract.} We prove that the empirical density of states
of quantum spin glasses on arbitrary graphs converges to a normal distribution
as long as the maximal degree is negligible compared
with the total number of edges. This extends the recent results of
\cite{Keating14} that
were proved for graphs with bounded chromatic number
and with symmetric coupling distribution.
Furthermore, we generalise the result to arbitrary hypergraphs. We test
the optimality of our condition on the maximal degree for $p$-uniform
hypergraphs that correspond to $p$-spin glass  Hamiltonians acting on $n$ distinguishable spin-$1/2$ particles. At the critical threshold
$p=n^{1/2}$ we find a sharp classical-quantum
phase transition between the normal distribution and the Wigner semicircle law.
The former is  characteristic 
to classical systems with commuting variables, while the latter is 
a signature of noncommutative 
random matrix theory.
\vspace{0.4cm}

\hspace{-1cm}\noindent\emph{Keywords.} Wigner semicircle law, Quantum spin glass, Sparse random matrix

\smallskip
\hspace{-1cm}\noindent{\bf AMS Subject Classification.} 15A52, 82D30
\end{abstract}

\section{Introduction}

The distribution of the energy levels for classical spin glasses 
converges to the normal distribution in the thermodynamic
limit by the central limit theorem.
On the other hand, the Hamiltonian of the quantum spin glasses can be considered 
as a random Hermitian matrix and thus the Wigner semicircle law 
might be expected.
In fact, the mean field quantum spin glass on the
full hypergraph with Gaussian coupling constants is equivalent
to the Gaussian unitary ensemble (GUE). It turns out that, despite
the inherent noncommutativity, the density of states for a large class of
quantum spin glasses still follows the normal law.
For quantum spin glasses on graphs with bounded chromatic number
and  with symmetrically distributed coupling constants
 this has recently been shown by Keating, Linden and Wells
\cite{Keating14}. In fact, their result extends 
to bounded deterministic couplings in case of spin chains,
see \cite{2014arXiv1403.1121K}.

In the first part of
this paper, we generalise their result in several directions
by considering general graphs and even hypergraphs. Moreover, we
relax the symmetry condition on the couplings.  We find
that the central limit theorem  holds
for a quantum spin glass on an arbitrary hypergraph, provided 
that the maximal degree of any vertex (the number of edges
adjacent to it) is much smaller than the total number of edges.
This condition guarantees that the noncommutative effects,
related to edges sharing a common vertex, are subleading:
most degrees of freedom are still commutative.
Thus the system is essential classical as far as the
density of states is concerned. We also present an example 
(star graph) where the degree of a distinguished vertex
is comparable with the total number of edges. The density
of states is explicitly computable and it is neither Gaussian nor
the semicircle law.

In the second part of the paper we investigate the transition
from the classical regime dominated by commuting variables
to the quantum regime where noncommutativity determines
the leading behaviour. This transition is particularly transparent
for the quantum $p$-spin model, i.e. a quantum spin glass on
a $p$-uniform hypergraph.  The case $p=2$ corresponds to
the quantum version of the standard Sherrington-Kirkpatrick 
 model and its density of states follows the normal law.  The other extreme
case,  $p=n$, is the GUE model with 
the semicircle law. We prove a sharp phase transition
at $p\sim\sqrt{n}$;  for $p\ll \sqrt{n}$ we have the normal law,
while for $p\gg \sqrt{n}$ we get the semicircle law. 
For $p =\lambda\sqrt{n}$ with a fixed $\lambda\in (0,\infty)$,
we establish a new family
of densities of states, parametrised by $\lambda$, that
naturally interpolates between the normal distribution and
the semicircle law. 
We
emphasise that the  regime  $p\sim n^\alpha$, $\alpha \in (1/2,1)$,
is still far from the mean field regime
in the sense of random matrices: we have only $3^p\binom{n}{p} \ll 2^n$
independent random variables parametrising an operator acting 
an $N=2^n$ dimensional
Hilbert space. In contrast, Wigner random matrices of dimension $N\times N$ 
have $N^2$ independent degrees of freedom. The $p$-spin model
thus corresponds to a very sparse random matrix, still it follows
the Wigner semicircle law if $p\gg \sqrt{n}$. Our result gives
a rigorous proof of the transition between 
the Gaussian and the semicircle  density of states that has been
numerically observed in \cite{French19715} for $k$-body interactions as $k$ approaches the total number of particles.

We mention that this phase transition is apparently present
only for the density of states; as far as the local eigenvalue statistics
is concerned all these models seem to belong to the random matrix (GUE) universality
class.  The  numerical tests
 presented in \cite{Keating14} deal with 
the one-dimensional quantum chain, 
one of the sparsest model, and still demonstrate
a very strong agreement with the GUE gap distribution.
Certainly the same is expected for spin glasses on denser graphs.
Quantum spin glasses are one of the simplest interacting
many-body disordered quantum models.
Therefore, this  remarkable feature is yet another  manifestation
of Wigner's vision on the ubiquity of the random matrix
gap statistics for essentially any disordered quantum system.
For more details on the physical motivation and related  works
we refer to \cite{Keating14}.

Our approach is different from that of \cite{Keating14}; we
use the very robust moment method.  In particular, this allows
us to consider arbitrary coupling distributions without much effort
and to identify new limiting laws 
in the $p\sim \sqrt{n}$ transition regime.

{\it Acknowledgement.} The authors are grateful to Jon Keating
for drawing their attention to the papers \cite{Keating14} and \cite{2014arXiv1403.1121K}.

\section{Model and main results}
\noindent Given a sequence of undirected graphs $\Gamma_n$ on the vertex sets $\{1,\dots,n\}$, we are considering Hermitian random matrices $H_n^{(\Gamma_n)}$ defined by 
\begin{equation}
H_n^{(\Gamma_n)}\defeq \frac{1}{\sqrt{9e(\Gamma_n)}} \sum_{(ij)\in\Gamma_n}\sum_{a,b=1}^3 \alpha_{a,b,(ij)}\sigma_i^{(a)}\sigma_{j}^{(b)},
\label{eq:Ham}
\end{equation}
where $e(\Gamma_n)$ denotes the number of edges in $\Gamma_n$. The normalisation factor of $(9e(\Gamma_n))^{-1/2}$ corresponds to the $9e(\Gamma_n)$ terms under the sum and is chosen to keep the spectrum of order $1$. As a convention, the edge connecting $i<j$ is called $(ij)$ and since the vertex set of the graphs is canonical, we shall, with a slight abuse of notation, identify the graph with its collection of edges. The coefficients $\alpha_{a,b,(ij)}$ are assumed to be independent random variables with zero mean and unit variance. The Pauli matrices acting on the $j$-th qubit are denoted by $\sigma_j^{(a)}\defeq 1_2^{\otimes(j-1)}\otimes\sigma^{(a)}\otimes 1_2^{\otimes(n-j)}$ where $\sigma^{(a)}$ are the standard spin-$1/2$  Pauli matrices
\[
\sigma^{(1)}=\begin{pmatrix}
0 &1\\
1&0	
\end{pmatrix}\qquad
\sigma^{(2)}=\begin{pmatrix}
0&-i\\
i&0	
\end{pmatrix}\qquad
\sigma^{(3)}=\begin{pmatrix}
	1&0\\0&-1
\end{pmatrix}
\]
and $1_2=\sigma^{(0)}$ is the $2\times 2$ identity matrix.  
For definiteness we will work with spin-$1/2$ systems, but all
our results hold for spin-$s$ models with any fixed $s$, see Remark \ref{rem:spins}. 

We are interested in the eigenvalue density of the operators $H_n^{(\Gamma_n)}$ in the limit $n\to\infty$. The expected eigenvalue density of $H_n^{(\Gamma_n)}$ is given by \[\mu_n \defeq \frac{1}{2^n}\E\sum_{j=1}^{2^n}\delta_{\lambda_j}, \] where $\lambda_j$ are the eigenvalues of $H_n^{(\Gamma_n)}$. The result in \cite{Keating14} shows that under the assumption that the random variables are bounded, symmetric about $0$ and that the graphs have a uniformly bounded chromatic number, $\mu_n$ converges weakly to a standard normal distribution as $n\to\infty$. 
Theorem \ref{thm:convergenceForGraphs} generalises this result by removing the symmetry condition and also allowing sequences of graphs for which the maximal vertex degree $d_\text{max}(n)$  grows slower than the number of edges $e(\Gamma_n)$. Since graph sequences with uniformly bounded chromatic numbers have a uniformly bounded maximal degree, our degree condition is implied by the condition from \cite{Keating14} on the chromatic number, but it is much more general, and in some sense optimal.

\begin{theorem}\label{thm:convergenceForGraphs}
Let $\Gamma_n$ be a sequence of graphs on the vertex sets $\{1,\dots,n\}$ such that $\lim_{n\to\infty}\frac{d_{\text{max}}(n)}{e(\Gamma_n)}=0$ and let 
\[
\Set{\alpha_{a,b,(ij)}|1\leq a,b\leq 3,~(ij)\in\Gamma_n} 
\] 
be a  tight collection of independent (not necessarily identically distributed) random variables with zero mean and unit variance. Then the expected density of states of the Hamiltonian $H_n^{(\Gamma_n)}$ defined in \eqref{eq:Ham} converges weakly to a standard normal distribution.
\end{theorem}

\begin{remark}
The empirical eigenvalue distribution $\nu_n\defeq 2^{-n}\sum_{j=1}^{2^n}\delta_{\lambda_j }$ is concentrated around its expectation $\mu_n=\E \nu_n$ and therefore the convergence in expectation, as proved in Theorem \ref{thm:convergenceForGraphs}, also implies that $\nu_n$ converges weakly in probability to a standard normal distribution. This strengthening of 
Theorem~\ref{thm:convergenceForGraphs} can be proved with a standard
extension of the moment method to estimating the variance
following the proof of \cite[Lemma 2.1.7]{anderson2010introduction}.
 Since  noncommutative features play no role in this argument, 
 we omit the details. The same remark also applies to our subsequent
 Theorems \ref{thm:pnSPINglass} and \ref{thm:convergenceForHypergraphs}.
\end{remark}

Theorem \ref{thm:convergenceForGraphs}
addresses both the model of nearest neighbour interactions in a $1$-dimensional closed chain (where the labelling is cyclic in the sense $\sigma_{n+1}^{(a)}=\sigma_1^{(a)}$)
\[H_n\defeq \frac{1}{\sqrt{9n}}\sum_{j=1}^n\sum_{a,b=1}^3\alpha_{a,b,j}\sigma_j^{(a)}\sigma_{j+1}^{(b)},\]
as well as the mean field model realised by the complete graph \[H_n^{(\text{comp})}\defeq\frac{1}{\sqrt{9n(n-1)/2}}\sum_{1\leq i<j\leq n}\sum_{a,b=1}^3 \alpha_{a,b,(ij)}\sigma_i^{(a)}\sigma_j^{(b)}.\]
 It also applies to all $d_n$-regular graphs in between, i.e. those where every vertex has the same degree $d_n\ge 1$ (here $(d_n)_{n\in\N}$ is an arbitrary sequence of parameters). Indeed, these graphs satisfy $n d_n=2e(\Gamma_n)$ and therefore
\[\frac{d_{\text{max}}(n)}{e(\Gamma_n)}=\frac{d_n}{e(\Gamma_n)}=\frac{2}{n}\to0\] as $n\to\infty$. 

We can generalise Theorem~\ref{thm:convergenceForGraphs} to hypergraphs allowing not only quadratic but also higher order spin interactions. The main condition is that the maximal {\it hyperedge degree}, i.e. the maximal number of hyperedges intersecting any fixed hyperedge, should be negligible compared with the total number of hyperedges. The precise formulation will be given in Theorem \ref{thm:convergenceForHypergraphs}, here we present only a prominent example of this generalisation, the quantum $p$-spin glasses. 
For any $p\ge 1$, the  Hamiltonian of a quantum $p$-spin glass  is given by \[H_n^{(p-\text{glass})}\defeq 3^{-p/2}\binom{n}{p}^{-1/2} \sum_{1\leq i_1<\dots<i_{p}\leq n}\sum_{a_1,\dots,a_{p}=1}^3\alpha_{a_1,\dots,a_{p},(i_1\dots i_{p})}\sigma_{i_1}^{(a_1)}\dots\sigma_{i_{p}}^{(a_{p})}. \] 
The following  theorem shows that the limiting density of states is Gaussian 
if $p$ is fixed or it is $n$-dependent, $p=p_n$, but grows slower than $\sqrt{n}$ i.e. $\lim_{n\to\infty}\frac{p_n}{\sqrt{n}}=0$. 
On the other hand, if $p_n$ grows faster than $\sqrt{n}$ i.e. $\lim_{n\to\infty}\frac{\sqrt{n}}{p_n}=0$, then the density of states
is given by the semicircle law. We shall use the notations $a_n\ll b_n$ and $a_n \gg b_n$ meaning that $\lim_{n\to\infty}\frac{a_n}{b_n}=0$ or $\lim_{n\to\infty}\frac{b_n}{a_n}=0$, respectively.

\begin{theorem}\label{thm:pnSPINglass}
Let $1\leq p_n\leq n$ be any sequence in $n$ and assume that the independent random variables $\alpha_{a_1,\dots,a_{p_n},(i_1\dots i_{p_n})}$ have zero mean, unit variance and form a tight family of random variables. Then the expected density of states of the Hamiltonians $H_n^{(p_n-\text{glass})}$ converges weakly to 
\begin{enumerate}[(i)]
\item a standard normal distribution if $p_n\ll \sqrt{n}$,
\item a semicircle distribution with density function $\rho(x)= \frac{1}{2\pi}\sqrt{4-x^2}\chi_{[-2,2]}(x)$ if $p_n\gg \sqrt{n}$,
\item  a distribution with the compactly supported density function 
\begin{align}\rho_\lambda(x)=\begin{cases}v(x|e^{-4\lambda/3})&\text{if }x\in \left[-\frac{2}{\sqrt{1-e^{-4\lambda/3}}},\frac{2}{\sqrt{1-e^{-4\lambda/3}}}\right], \\ 0 &\text{else}\end{cases}\label{eq:rhoLambda}\end{align} where 
\[v(x|q)\defeq\frac{\sqrt{1-q}}{\pi\sqrt{1-(1-q)x^2/4}}\prod_{k=0}^\infty \left[\frac{1-q^{2k+2}}{1-q^{2k+1}}\left(1-\frac{x^2(1-q)q^k}{(1+q^k)^2}\right)\right]\] if $\lim_{n\to\infty}\frac{p_n}{\sqrt{n}}=\lambda$.
\end{enumerate}
\end{theorem}

\section{Moment Method}
\noindent To show that the expected density of states converges weakly to some distribution $\mu$ it often suffices to show that the moments \[m_{k,n}\defeq\int_\R x^n\diff\mu_n(x)=\E \frac{1}{2^n}\Tr (H_n^{(\Gamma_n)})^k=\frac{1}{2^n}\Tr\E (H_n^{(\Gamma_n)})^k\] of $\mu_n$ converge pointwise to the moments $m_k\defeq \int_\R x^k\diff\mu(x)$ of $\mu$ as $n\to\infty$. A sufficient condition for the uniqueness of the limiting distribution is given by Carleman's condition (see \cite{carleman1926fonctions}): A probability distribution $\mu$ is uniquely determined by its moments $m_k$ if $\sum_{k=1}^\infty m_{2k}^{-1/2k}=\infty$.

To keep the terms simple we introduce the notations \[\sigma_J\defeq \sigma_i^{(a_1)}\sigma_{j}^{(a_2)},\qquad \alpha_J\defeq\alpha_{a_1,a_2,(ij)}\] for tuples $J=({\bf a},(ij))=(a_1,a_2,(ij))$ and denote the index sets by \[I_n\defeq \{1,2,3\}^2\times \Gamma_n =\Set{({\bf a},(ij))=(a_1,a_2,(ij))| {\bf{a}}=(a_1,a_2)\in\{1,2,3\}^2,(ij)\in\Gamma_n}.\]In order to compute the $k$-th moment $m_{k}$ we have to evaluate the sum \begin{align} m_{k,n}&=2^{-n}\Tr\E (H_n^{(\Gamma_n)})^k=2^{-n}\Tr\E \Big(\frac{1}{\sqrt{9e(\Gamma_n)}}\sum_{J\in I_n} \sigma_J \alpha_J\Big)^k\nonumber\\&=(9e(\Gamma_n))^{-k/2}\sum_{J_1,\dots,J_k\in I_n}2^{-n}\Tr\sigma_{J_1}\dots\sigma_{J_k}\E\alpha_{J_1}\dots\alpha_{J_k} \label{eq:expansion}\end{align} in the limit $n\to\infty$. We split the sum
in \eqref{eq:expansion}  into
 three disjoint  parts 
\begin{equation}\label{eq:threesplit}
\sum_{J_1,\dots,J_k\in I_n} = \sum_{D_{n,k}} +  \sum_{A_{n,k}} + \sum_{B_{n,k}} 
\end{equation}
for a partition $I_n^k = A_{n,k}\cup B_{n,k}\cup D_{n,k}$ into three 
subsets defined below. 

Given any set $X$ and any integer $k\ge 1$ we define $P_2(X^k)$ to be those tuples $(x_1,\dots,x_k)\in X^k$ for which all entries $x_1,\dots,x_k$ appear exactly twice. Firstly, we split $I_n^k$ into the disjoint sets $P_2(I_n^k)$ and its complement $D_{n,k}\defeq I_n^k\setminus P_2(I_n^k)$ and then further split $P_2(I_n^k)$ into 
\begin{align*}
A_{n,k}\defeq \Set{(({\bf a}_1,e_1),\dots,({\bf a}_k,e_k))\in P_2(I_n^k)| e_i\cap e_j=\emptyset\text{ if } e_i\not=e_j}, \end{align*} 
the family of $k$-tuples with all edges non-intersecting, and 
its complement $B_{n,k}\defeq P_2(I_n^k)\setminus A_{n,k}$. The condition $e_i\cap e_j=\emptyset$ (meaning that the edges have no vertex in common) assures that the matrices $\sigma_{({\bf a}_i,e_i)}$ and $\sigma_{({\bf a}_j,e_j)}$ commute. The reasons for these two splits are of entirely different nature. As Lemma \ref{lemma:AnkSUM} below shows the sum over $D_{n,k}$ is negligible under fairly general circumstances due to combinatorics without using any properties of the traces but their boundedness. The second split of the remaining $P_2(I_n^k)$ into $A_{n,k}\cup B_{n,k}$ is important since for $(J_1,\dots,J_k)\in A_{n,k}$ the $\sigma_{J_i}$ corresponding to different $J_i$ commute and the can be reordered in such a way that only squares of Pauli matrices remain and the normalised trace is $1$. That means that all relevant quantum effects due to (potential) non-commutativities are isolated in the index set $B_{n,k}$. The system is essentially classical if the contribution of the index set $B_{n,k}$ to the rhs. of eq. (\ref{eq:expansion}) can be neglected. If this is the case, the asymptotic eigenvalue distribution equals the asymptotic energy histogram of the corresponding classical model where the spin matrices are replaced by commuting spins $s_i\in\{-1,1\}$ or $s_i\in S^2$.
\begin{lemma}\label{lemma:AnkSUM}
Let $(X_n)_{n\in\N}$ be a growing sequence of index sets and let $\Set{\alpha_x | n\in\N, x\in X_n}$ be a family of independent random variables with zero mean and unit variance and uniformly bounded moments $\abs{\E\alpha_x^k}\leq C_k<\infty$. Then we have the bound
\begin{align*}
\abs{X_n}^{-k/2}&\sum_{(x_1,\dots,x_k)\in X_n^k\setminus P_2(X_n^k)}\left\lvert\E\alpha_{x_1}\dots\alpha_{x_k}\right\lvert\\&\leq \begin{cases} 
0 \quad&\text{if } k<3\text{ is odd,}\\
\frac{k!!}{\abs{X_n}^{1/2}}\frac{ (k-1)C_3}{3} +\O{\abs{X_n}^{-3/2}} &\text{if }k\geq 3\text{ is odd,}\\ 
0&\text{if }k<4\text{ is even,}\\
\frac{(k-1)!!}{\abs{X_n}}\left(\frac{ k(k-2)(k-4)C_3^2}{18}+\frac{k(k-2)C_4}{4}\right)+\O{\abs{X_n}^{-2}} &\text{if }k\geq4\text{ is even}
\end{cases} 
\end{align*} 
as $n\to\infty$ while $k$ is fixed.
\end{lemma}
\begin{proof}
First note that in the case that some $x_i$ only appears once, by independence and zero mean hypothesis these terms vanish identically. 
 Since the case that all $x_i$'s appear exactly twice is excluded from the index set in the sum above,
 we only have to consider those terms for which there are strictly less than $\frac{k}{2}$ distinct $x_i$'s. There are only $\O{\abs{X_n}^m}$ (as $n\to\infty$) terms with $m<\frac{k}{2}$ distinct $x_i$'s,
so we find,  after summation,  that the total contribution of these terms vanish as $\O{\abs{X_n}^{m-k/2}}$ as $n\to \infty$. 
Let us try to find the coefficient of the highest order term in $n$. In the case that $k$ is odd the term with the highest order comes from $m=\frac{k-1}{2}$ i.e. vanishes for $k<3$. The terms with $\frac{k-1}{2}$ distinct $x_i$'s such that all $x_i$ appear at least twice are those for which some $x_i$ appears three times and the rest only two times. 
There are $\binom{n}{k/2-1/2}$ ways of choosing the $x_i$'s to appear, then there are $\frac{k-1}{2}$ ways of choosing the $x_i$ that appears three times and there are $\frac{k!}{3\cdot2^{k/2-3/2}}$ ways of assigning those pre-described values to the tuples $(x_1,\dots,x_k)$. In total we find for the number of terms contributing to the leading order
\[\binom{\abs{X_n}}{ k/2-1/2}\frac{k-1}{2}\frac{k!}{3\cdot 2^{k/2-3/2}}=\frac{\abs{X_n}^{k/2-1/2}k!}{3(k-3)!!}+\O{\abs{X_n}^{k/2-3/2}},\] each having a modulus bounded by $C_3$. The proof for even $k$ is analogous.
\end{proof}

The estimate in Lemma~\ref{lemma:AnkSUM} used only the scaling properties of the expectations and applies to the computation of $m_{k,n}$ since the normalised traces have uniformly bounded modulus. By using more specifics of the tracial part we can improve the error estimate from Lemma \ref{lemma:AnkSUM} significantly:

\begin{lemma}\label{lemma:AnkSUMwithTraces}
Assume that the random variables $\alpha_J$ are independent, have zero mean, unit variance and uniformly bounded moments $\abs{\E\alpha_J^k}\leq C_k<\infty$ for all $n\in\N, J\in I_n$. Then we have the bound 
\begin{align*}
\Big\lvert(9e(\Gamma_n))^{-k/2} &\sum_{(J_1,\dots,J_k)\in D_{n,k}}2^{-n}\Tr \sigma_{J_1}\dots\sigma_{J_k}\E\alpha_{J_1}\dots\alpha_{J_k}\Big\lvert\\ &\leq \begin{cases} 
0 &\text{if } k<9 \text{ is odd}\\
\frac{k!!}{e(\Gamma_n)^{3/2}}\frac{C_3^3(k-1)(k-3)(k-5)(k-7)}{3^{5}}+\O{ e(\Gamma_n)^{-5/2}}  &\text{if } k\geq 9 \text{ is odd}\\
0 &\text{if } k<4 \text{ is even}\\
\frac{(k-1)!!}{e(\Gamma_n)}\frac{k(k-2)C_4}{36}+\O{e(\Gamma_n)^{-2}}&\text{if }k\geq 4\text{ is even}
\end{cases} 
\end{align*} 
as $n\to\infty$ while $k$ is fixed. 
\end{lemma}

Before going into the proof of Lemma \ref{lemma:AnkSUMwithTraces} we state some properties of the traces of Pauli matrices we shall need. A proof of this technical Lemma is given in the appendix.

\begin{lemma}[Traces of products of Pauli matrices]\label{lemma:pauliMatrices}
Given $a_1,\dots,a_k\in\{1,2,3\}$ the normalised traces of products of Pauli matrices $\sigma(a_1,\dots,a_k)\defeq\frac{1}{2}\Tr\sigma^{(a_1)}\dots\sigma^{(a_k)}$ satisfy:
\begin{enumerate}[(i)]
\item $\sigma(a_1,\dots,a_k)\in\{0,1,-1,i,-i\}$;
\item More generally for all $1\leq j_1,\dots,j_k\leq n$\[\frac{1}{2^n}\Tr \sigma_{j_1}^{(a_1)}\dots\sigma_{j_k}^{(a_k)}\in\{0,1,-1,i,-i\};\]
\item $\sigma(a_1,\dots,a_k)$ is non-zero if and only if the parities of the numbers of $1$'s, $2$'s and $3$'s among the $a_1,\dots,a_k$ coincide.
\item If $k$ is even we have the recursion relation \[\sigma(a_1,\dots,a_k)=\sum_{j=2}^k \delta_{a_1a_j}(-1)^j\sigma(a_2,\dots,\widehat{a_j},\dots,a_k) \] where $\widehat{a_j}$ means that the $j$-th entry is omitted;
\end{enumerate} 
\end{lemma}

\begin{proof}[Proof of Lemma \ref{lemma:AnkSUMwithTraces}]
For odd $k$ first note that up to a factor of $\pm 1$ we can reorder the $\sigma_{J_i}$'s in the expression $2^{-n}\Tr\sigma_{J_1}\dots\sigma_{J_k}$ since Pauli matrices either commute or anti-commute. Since for any $J_i$ the we have $\sigma_{J_i}^2=1_{2^n}$ the normalised trace reduces to $\pm 2^{-n}\Tr \sigma_{J_{i_1}}\dots \sigma_{J_{i_l}}$ with the $J_{i_1},\dots,J_{i_l}$ being all exactly those distinct $J_1,\dots,J_k$ that appear an odd number of times. By Lemma \ref{lemma:pauliMatrices}(iii) in each component of the tensor product we get a zero trace if there are either one or two different Pauli matrices acting on it. Hence the normalised trace is zero if there are one or two distinct $J_i$ appearing an odd number of times. Since $k$ is odd we therefore see that the highest order contribution comes from the term where three distinct $J_i$ appear three times and the rest appear two times. Thus for $k<9$ we see that $m_{k,n}$ is identically zero. For $k\geq 9$ we first choose the $\left(\frac{k-9}{2}+3\right)=\frac{k-3}{2}$ distinct $J_i$ to appear and then those three to appear three times. By counting the number of ways of assigning those $J_i$ to our tuples we therefore find the factor  \begin{align*}
\binom{9e(\Gamma_n)}{k/2-3/2}\binom{k/2-3/2}{3}\frac{k!}{3^3\cdot 2^{k/2-3/2}}=\frac{(9e(\Gamma_n))^{k/2-3/2}}{(k/2-3/2)!}\binom{k/2-3/2}{3}\frac{k!}{3^3\cdot 2^{k/2-3/2}}+\O{(e(\Gamma_n))^{k/2-5/2}}
\end{align*} as $n\to\infty$ from which after dividing by $(9e(\Gamma_n))^{k/2}$ the claimed asymptotics follow. The claims for even $k$ immediately follow from Lemma \ref{lemma:AnkSUM} using that the
term with $C_3^2$ vanishes by the above argument (since there are two distinct $J_i$ appearing an odd number of times).
\end{proof}

In particular this already shows that in the limit $n\to\infty$ all odd moments vanish. The situation with the sums over $A_{n,k}$ and $B_{n,k}$ in \eqref{eq:threesplit} is a little bit more delicate. For a large class of graph sequences the sum over $B_{n,k}$ is also negligible and the only contribution comes from $A_{n,k}$ where all normalised traces are equal to $1$ and the system is essentially classical. In this case the non-commutativity is actually only a small perturbation and consequently we see the same result as in the classical central limit theorem rather than a random matrix semicircle law. 
We are now ready to give a proof of Theorem \ref{thm:convergenceForGraphs}, including explicit estimates regarding the rate of convergence of the moments.

\begin{reptheorem}{thm:convergenceForGraphs}[Detailed version]
Denote the maximal vertex degree in the graph $\Gamma_n$ by $d_\text{max}(n)$. Let $\Gamma_n$ be a sequence of graphs on the vertex sets $\{1,\dots,n\}$ such that $\lim_{n\to\infty}\frac{d_{\text{max}}(n)}{e(\Gamma_n)}=0$ and let \[\Set{\alpha_{a,b,(ij)}|1\leq a,b\leq 3,~(ij)\in\Gamma_n} \] be a tight collection of independent (not necessarily identically distributed) random variables with zero mean and unit variance. Then the Hamiltonian defined 
by \[H_n^{(\Gamma_n)}\defeq \frac{1}{\sqrt{9e(\Gamma_n)}}\sum_{(ij)\in\Gamma_n}\sum_{a,b=1}^3\alpha_{a,b,(ij)}\sigma_i^{(a)}\sigma_{j}^{(b)}\] (where as a convention the edge between $i<j$ is denoted by $(ij)$) has an expected density of states which converges weakly to a standard normal distribution. The convergence rate of the moments is of order $e(\Gamma_n)^{-3/2}$ for odd moments and $\frac{d_\text{max}(n)}{e(\Gamma_n)}$ for even moments. Moreover $m_{k,n}\equiv m_k$ for $k\in\{0,1,2,3,5,7\}$.
\end{reptheorem}

\begin{proof}
First note that we can, without loss of generality, assume that the random variables are uniformly bounded and therefore have also uniformly bounded moments. This follows from a standard reduction step relying on the Hoffman-Wielandt inequality that allows us to approximate the density of states $\mu_n$ by the density of states of a Hamiltonian with truncated random variables. For details the reader is referred to the proof of \cite[Theorem 2.1.21]{anderson2010introduction} which can be adapted to our model by the tightness assumption on the random variables. This reduction step is also valid in the proofs
 of Theorems \ref{thm:pnSPINglass}, \ref{thm:convergenceForHypergraphs} and Proposition \ref{prop:StarGraph},  where we shall assume it without further explanation. 

The treatment of the sum from eq.~\eqref{eq:expansion} is performed in three steps according to the split from eq. \eqref{eq:threesplit}. Lemma~\ref{lemma:AnkSUMwithTraces} dealt with the $D_{n,k}$-part of the sum. We now consider the part of the sum over the index set $B_{n,k}$. From the condition $\frac{d_\text{max}(n)}{e(\Gamma_n)}\to 0$ as $n\to\infty$ it follows that the number $d_{j,n}$ of choosing $j$ non intersecting edges from the graph $\Gamma_n$ asymptotically behaves as $e(\Gamma_n)^j/j!$ i.e. $\lim_{n\to\infty}\frac{d_{j,n}}{e(\Gamma_n)^j/j!}=1$ for all fixed $j$. Indeed, there are $e(\Gamma_n)$ choices for the first edge $(il)$. For the next edge we can pick all edges except those including $i$ and $l$ i.e. there are at least $e(\Gamma_n)-2d_{\text{max}}(n)$ choices for the second edge. Continuing we find the bound 
\begin{align}\frac{e(\Gamma_n)^j}{j!}&\geq d_{j,n}\geq \frac{1}{j!}e(\Gamma_n)(e(\Gamma_n)-2d_{\text{max}}(n))((e(\Gamma_n)-4d_{\text{max}}(n)))\dots(e(\Gamma_n)-2(j-1)d_{\text{max}}(n))\nonumber\\
&= \frac{e(\Gamma_n)^j}{j!} \left(1- \frac{j(j-1)d_{\text{max}}(n)}{e(\Gamma_n)} +\O{\left(\frac{d_{\text{max}}(n)}{e(\Gamma_n)}\right)^2}\right)\label{eq:nonIntersectingEdges}\end{align}
 as $n\to\infty$. Dividing by $e(\Gamma_n)^j/j!$ then proves that $d_{j,n}$ asymptotically behaves as $e(\Gamma_n)^j/j!$. 

 The estimate from eq.~\eqref{eq:nonIntersectingEdges} also shows that the number of choosing $j$ edges that have at least one intersection is, to leading order, at most given by \[\frac{e(\Gamma_n)^{j-1}d_\text{max}(n)}{(j-2)!}+\O{\left(\frac{d_{\text{max}}(n)}{e(\Gamma_n)}\right)^2 e(\Gamma_n)^j}.\] Since there are $\binom{k}{2}\binom{k-2}{2}\dots\binom{2}{2}=\frac{k!}{2^{k/2}}$ ways of assigning $\frac{k}{2}$ chosen edges to $e_1,\dots,e_k$ such that each appears twice, the index set $B_{n,k}$ therefore contains at most \[9^{k/2}\frac{k!}{2^{k/2}}\frac{d_\text{max}(n)}{(k/2-2)!}e(\Gamma_n)^{k/2-1}+\O{\left(\frac{d_{\text{max}}(n)}{e(\Gamma_n)}\right)^2 e(\Gamma_n)^{k/2}}\] 
elements as $n\to\infty$. Using that the modulus of the normalised traces is at most $1$ (see Lemma \ref{lemma:pauliMatrices}) and that the expectations are all equal to $1$ due to unit variance and independence, we therefore found the bound
\begin{align}\label{eq:SUMBnk}
\left\lvert\frac{1}{(9e(\Gamma_n))^{k/2}}\sum_{(J_1,\dots,J_k)\in B_{n,k}}2^{-n}\Tr\sigma_{J_1}\dots\sigma_{J_k}\E\alpha_{J_1}\dots\alpha_{J_k}\right\lvert\nonumber\\\leq (k-1)!!\frac{k(k-2)}{4}\frac{d_{\text{max}}(n)}{e(\Gamma_n)}+\O{\left(\frac{d_\text{max}(n)}{e(\Gamma_n)}\right)^2}
\end{align} as $n\to\infty$ while $k$ is fixed.

For the summation over $A_{n,k}$ in~\eqref{eq:threesplit}, we first note that all terms under the sum are equal to $1$. Indeed, the expectations are again $1$ by independence and unit variance. For the traces we find that since all distinct $J_i$ act on distinct qubits,  in all components there is either an identity matrix or a product of two identical Pauli matrices i.e. again identity matrices. Again similarly to eq.~\eqref{eq:nonIntersectingEdges} we can estimate $\abs{A_{n,k}}$ to get 
\begin{align}
(k-1)!!&\geq(9e(\Gamma_n))^{-k/2}\abs{A_{n,k}}\geq (k-1)!!\left(1-\frac{k(k-2)}{4}\frac{d_{\text{max}}(n)}{e(\Gamma_n)}\right)+\O{\left(\frac{d_\text{max}(n)}{e(\Gamma_n)}\right)^2}\label{eq:SUMCnk}
\end{align}	
as $n\to\infty$ while $k$ is fixed. 

The moments $m_k$ of a standard normal distribution are given by $(k-1)!!$ for even $k$ and $0$ for odd $k$ and satisfy Carleman's continuity condition. Using the bound in Lemma \ref{lemma:AnkSUMwithTraces} together with eqs.~\eqref{eq:SUMBnk} and~\eqref{eq:SUMCnk} we arrive at 
\begin{align*}
\abs{m_{k,n}-(k-1)!!}\leq& \left\lvert(9e(\Gamma_n))^{-k/2}\sum_{(J_1,\dots,J_k)\in D_{n,k}}2^{-n}\Tr\sigma_{J_1}\dots\sigma_{J_k}\E\alpha_{J_1}\dots\alpha_{J_k}\right\lvert \\
&+\left\lvert(9e(\Gamma_n))^{-k/2}\sum_{(J_1,\dots,J_k)\in B_{n,k}}2^{-n}\Tr\sigma_{J_1}\dots\sigma_{J_k}\E\alpha_{J_1}\dots\alpha_{J_k}\right\lvert\\
&+\left\lvert(9e(\Gamma_n))^{-k/2}\sum_{(J_1,\dots,J_k)\in A_{n,k}}2^{-n}\Tr\sigma_{J_1}\dots\sigma_{J_k}\E\alpha_{J_1}\dots\alpha_{J_k}-(k-1)!!\right\lvert\\ 
\leq&(k-1)!!k(k-2)\left(\frac{C_4}{36 e(\Gamma_n)}+\frac{d_\text{max}(n)}{2e(\Gamma_n)}\right)+\O{\left(\frac{d_\text{max}(n)}{e(\Gamma_n)}\right)^2}
\end{align*} for (fixed) even $k\geq 9$ as $n\to\infty$, whereas
\begin{align*}\abs{m_{k,n}-0}&= \left\lvert(9e(\Gamma_n))^{-k/2}\sum_{(J_1,\dots,J_k)\in D_{n,k}}2^{-n}\Tr\sigma_{J_1}\dots\sigma_{J_k}\E\alpha_{J_1}\dots\alpha_{J_k}\right\lvert \\&\leq \frac{k!!}{e(\Gamma_n)^{3/2}}\frac{C_3^3(k-1)(k-3)(k-5)(k-7)}{3^{5}}+\O{e(\Gamma_n)^{-5/2}} \end{align*}
for (fixed) odd $k$ as $n\to\infty$. This shows the convergence of each moment, thus the weak convergence. The claim that the moments $m_{k,n}$ agree identically with $m_k$ for $k\in\{1,2,3,5,7\}$ follows immediately from Lemma \ref{lemma:AnkSUMwithTraces}.
\end{proof}

As the following example shows, the assumption on the growth of the maximal degree is  necessary. Let 
\[H_n^{\text{(star)}}\defeq\frac{1}{\sqrt{9(n-1)}}\sum_{j=2}^n\sum_{a,b=1}^3 \alpha_{a,b,j}\sigma_1^{(a)}\sigma_j^{(b)} \] 
be the Hamiltonian corresponding to the star graph in which, say, the vertex $1$ is connected to all other vertices while there are no edges between the rest. This model shows a significantly different limiting behaviour (a proof is given in the Appendix, see also Figure \ref{fig:NumericalDOS}):

\begin{figure}
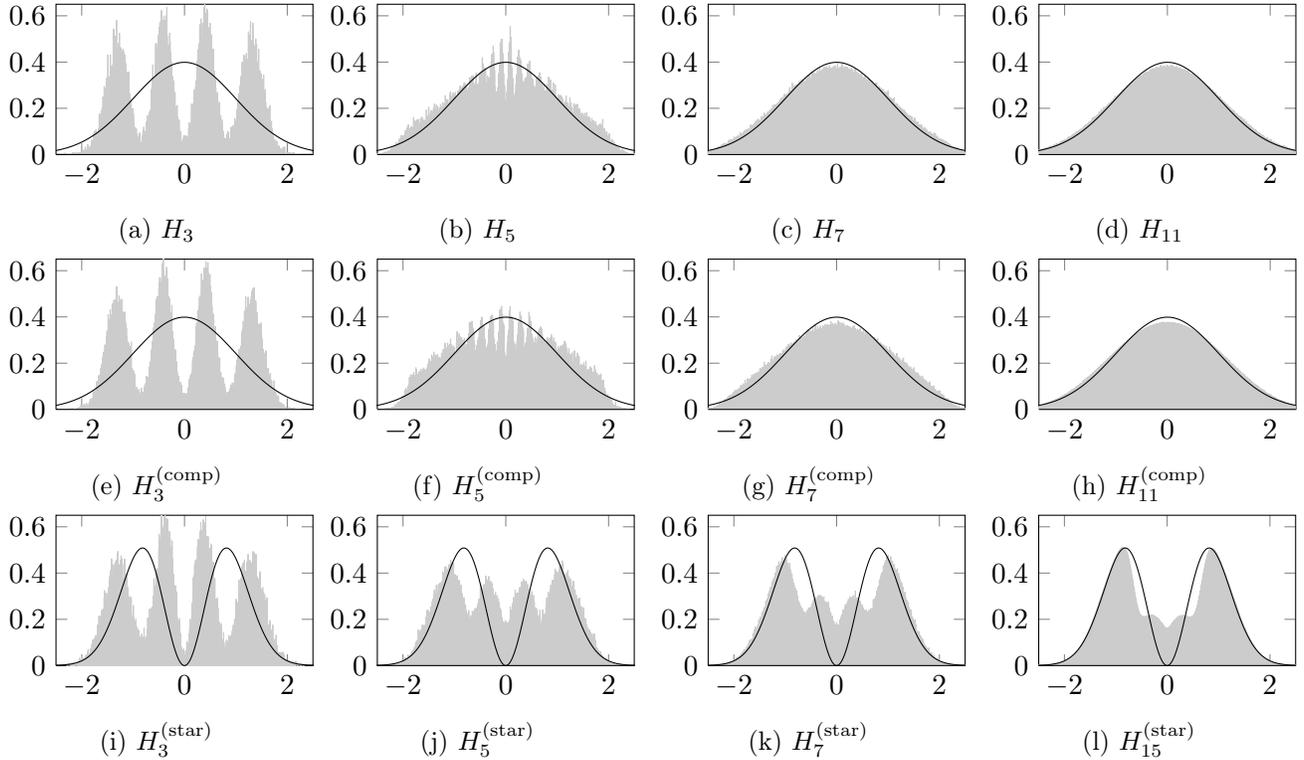

\begin{subfigure}{.24\textwidth}
  \centering
  \newlength\figureheight \newlength\figurewidth \setlength\figureheight{2cm} \setlength\figurewidth{.8\linewidth} \input{L3.tikz}
  \caption{$H_3$}
\end{subfigure}%
\begin{subfigure}{.24\textwidth}
  \centering
  \setlength\figureheight{2cm} \setlength\figurewidth{.8\linewidth}\input{L5.tikz}
  \caption{$H_5$}
\end{subfigure}
\begin{subfigure}{.24\textwidth}
  \centering
  \setlength\figureheight{2cm} \setlength\figurewidth{.8\linewidth}\input{L7.tikz}
  \caption{$H_7$}
\end{subfigure}
\begin{subfigure}{.24\textwidth}
  \centering
   \setlength\figureheight{2cm} \setlength\figurewidth{.8\linewidth}\input{L11.tikz}
  \caption{$H_{11}$}
\end{subfigure}\\
\begin{subfigure}{.24\textwidth}
  \centering
  \setlength\figureheight{2cm} \setlength\figurewidth{.8\linewidth} \input{L3COMP.tikz}
\caption{$H_3^{\text{(comp)}}$}
\end{subfigure}%
\begin{subfigure}{.24\textwidth}
  \centering
  \setlength\figureheight{2cm} \setlength\figurewidth{.8\linewidth}\input{L5COMP.tikz}
\caption{$H_5^{\text{(comp)}}$}
\end{subfigure}
\begin{subfigure}{.24\textwidth}
  \centering
  \setlength\figureheight{2cm} \setlength\figurewidth{.8\linewidth}\input{L7COMP.tikz}
\caption{$H_7^{\text{(comp)}}$}
\end{subfigure}
\begin{subfigure}{.24\textwidth}
  \centering
   \setlength\figureheight{2cm} \setlength\figurewidth{.8\linewidth}\input{L11COMP.tikz}
  \caption{$H_{11}^{\text{(comp)}}$}
\end{subfigure}\\
\begin{subfigure}{.24\textwidth}
  \centering
  \setlength\figureheight{2cm} \setlength\figurewidth{.8\linewidth} \input{L3STAR.tikz}
\caption{$H_3^{\text{(star)}}$}
\end{subfigure}%
\begin{subfigure}{.24\textwidth}
  \centering
  \setlength\figureheight{2cm} \setlength\figurewidth{.8\linewidth}\input{L5STAR.tikz}
\caption{$H_5^{\text{(star)}}$}
\end{subfigure}
\begin{subfigure}{.24\textwidth}
  \centering
  \setlength\figureheight{2cm} \setlength\figurewidth{.8\linewidth}\input{L7STAR.tikz}
\caption{$H_7^{\text{(star)}}$}
\end{subfigure}
\begin{subfigure}{.24\textwidth}
  \centering
   \setlength\figureheight{2cm} \setlength\figurewidth{.8\linewidth}\input{L14STAR.tikz}
  \caption{$H_{15}^{\text{(star)}}$}
\end{subfigure}
\caption{Empirical density of states for sample size $1000$ including the limiting density}
\label{fig:NumericalDOS}
\end{figure}

\begin{prop}\label{prop:StarGraph}
Suppose that the random variables $\Set{\alpha_{a,b,j}|n\in\N, 1\leq a,b\leq 3, 2 \leq j\leq n}$ are independent, have zero mean, unit variance and form a tight family of random variables. The expected density of states of $H_n^{(\text{star})}$ then converges weakly to a distribution with density 
\[ \rho(x)= 3\sqrt{\frac{3}{2\pi}}x^2 e^{-3x^2/2}\] as $n\to\infty$.
\end{prop}

\section{Hypergraphs}
A hypergraph is a generalised graph in which any hyperedge can contain a variable number of vertices. Formally a hypergraph on a vertex set $V$ is any subset of $\mathcal{P}(V)\setminus\emptyset$. A hyperedge $e$ containing the (distinct) vertices $i_1<\dots< i_l$ will be denoted by $(i_1\dots i_l)$. We shall use the notation $\abs{e}\defeq l$ for the number of vertices in a given hyperedge $e=(i_1\dots i_l)$. Just as in the traditional graph, the degree of a vertex is defined to be the number of hyperedges containing the given vertex. The total number of hyperedges is again denoted by $e(\Gamma_n)$. For a given hypergraph $\Gamma_n$ on the vertex set $\{1,\dots,n\}$ we introduce the notations \[\alpha_J\defeq \alpha_{({\bf a},e)}\defeq \alpha_{a_1,\dots,a_l,(i_1\dots i_l)},
\qquad \sigma_J\defeq \sigma_{({\bf a},e)}\defeq \sigma_{i_1}^{(a_1)}\dots\sigma_{i_l}^{(a_l)}\] and $\abs{J}=\abs{e}=l$ for \[J=({\bf a},e)=(a_1,\dots,a_l,(i_1\dots i_l))\in I_n\defeq \Set{({\bf a},e) | e\in \Gamma_n, {\bf a}\in \{1,2,3\}^{\abs{e}}}. \] The generalised Hamiltonian corresponding to the hypergraph $\Gamma_n$ is defined to be
\begin{align} H_n^{(\Gamma_n)}\defeq \frac{1}{\sqrt{e(\Gamma_n)}}\sum_{e\in\Gamma_n}\frac{1}{3^{\abs{e}/2}}\sum_{{\bf a}\in \{1,2,3\}^{\abs{e}}}\alpha_{({\bf a},e)}\sigma_{({\bf a},e)}.\label{eq:HnHypergraph}\end{align}
We again want to study the moments 
\begin{align}m_{k,n}=2^{-n}\Tr\E (H_n^{(\Gamma_n)})^k=(e(\Gamma_n))^{-k/2}\sum_{J_1,\dots,J_k\in I_n}3^{-(\abs{J_1}+\dots+\abs{J_k})/2}2^{-n}\Tr\sigma_{J_1}\dots\sigma_{J_k}\E\alpha_{J_1}\dots\alpha_{J_k} \label{eq:expansionHypergraph}\end{align} in the limit $n\to\infty$. Lemma \ref{lemma:AnkSUM} again applies and immediately shows that we can restrict our attention to those summands where the $J_1,\dots,J_k$ appear in pairs of two. If the hyperedges of the $\frac{k}{2}$ distinct $J_i$'s are disjoint we can reorder the $\sigma_{J_i}$'s freely and therefore get a normalised trace of $1$. As in the proof of Theorem \ref{thm:convergenceForGraphs} we establish a sufficient criterion on the sequence of graphs 
such that among all families of   $\frac{k}{2}$ edges the proportion of those that have
 mutually disjoint edges approaches $1$. 

As for conventional graphs, the \emph{line graph} $L(\Gamma_n)$ of a hypergraph $\Gamma_n$ is graph whose vertices are the hyperedges $\{e_1,\dots,e_M\}$ of $\Gamma_n$. Two vertices of $L(\Gamma_n)$ (i.e. hyperedges of $\Gamma_n$) $e_1, e_2$ are adjacent (connected
by an edge in the line graph) if and only if $e_1$ and $e_2$ are non-disjoint and so the edges of $L(\Gamma_n)$ are given by \[\Set{(e_i e_j) | 1\leq i,j\leq M, e_i\cap e_j\ne \emptyset}.\] Given some fixed hyperedge $e_1\in \Gamma_n$ there are at least $e(\Gamma_n)-d_\text{max}(L(\Gamma_n))$ hyperedges $e_2$ disjoint from $e_1$, where the maximal hyperedge degree $d_\text{max}^{(e)}(n)\defeq d_\text{max}(L(\Gamma_n))$ is the maximal vertex degree of the line graph. Continuing we find for the number $d_{j,n}$ of choices of $j$ disjoint hyperedges from $\Gamma_n$ the bound
\begin{align}
\frac{e(\Gamma_n)^j}{j!}&\geq d_{j,n}\geq \frac{1}{j!}e(\Gamma_n)(e(\Gamma_n)-d_\text{max}^{(e)}(n))\dots(e(\Gamma_n)-(j-1)d_\text{max}^{(e)}(n))\nonumber\\
&=\frac{e(\Gamma_n)^j}{j!}\left(1-\frac{j(j-1)}{2}\frac{d_\text{max}^{(e)}(n)}{e(\Gamma_n)}+\O{\left(\frac{d_\text{max}^{(e)}(n)}{e(\Gamma_n)}\right)^2}\right)
\end{align} as $n\to \infty$ while $j$ is fixed if $\lim_{n\to\infty}\frac{d_\text{max}^{(e)}(n)}{e(\Gamma_n)}=0$. Following the proof of Theorem \ref{thm:convergenceForGraphs} we therefore proved its generalisation for hypergraphs:

\begin{theorem}\label{thm:convergenceForHypergraphs}
Let $\Gamma_n$ be a sequence of graphs on the vertex sets $\{1,\dots,n\}$ such that $\lim_{n\to\infty}\frac{d_{\text{max}}^{(e)}(n)}{e(\Gamma_n)}=0$ and let \[\Set{\alpha_{({\bf a},e)}| e\in \Gamma_n, {\bf a}\in\{1,2,3\}^{\abs{e}} }\] be a  tight  collection of independent (not necessarily identically distributed) random variables with zero mean and unit variance. Then the Hamiltonian defined in~\eqref{eq:HnHypergraph} has a density of states which converges weakly to a standard normal distribution.
\end{theorem}
For $2$-uniform hypergraphs (meaning that all edges connect $2$ vertices) the statement of this Theorem is equivalent to Theorem \ref{thm:convergenceForGraphs}. More generally the theorem also covers a sequence of $p_n$-uniform graphs $\Gamma_n$ corresponding to the $p_n$-spin glasses. An interesting special case is the sequence of complete $p_n$-uniform hypergraphs in which the hyperedges connect any $p_n$ distinct vertices. The corresponding Hamiltonians are given by \[H_n^{(p_n-\text{glass})}\defeq 3^{-p_n/2}\binom{n}{p_n}^{-1/2} \sum_{1\leq i_1<\dots<i_{p_n}\leq n}\sum_{a_1,\dots,a_{p_n}=1}^3\alpha_{a_1,\dots,a_{p_n},(i_1\dots i_{p_n})}\sigma_{i_1}^{(a_1)}\dots\sigma_{i_{p_n}}^{(a_{p_n})}.\] In this case the degree of any hyperedge is \[\deg(i_1\dots i_{p_n})=\binom{n}{p_n}-\binom{n-p_n}{p_n}, \] while the total number of hyperedges is given by $e(\Gamma_n)=\binom{n}{p_n}$. Since \[\lim_{n\to\infty} \frac{\binom{n-p_n}{p_n}}{\binom{n}{p_n}}=
\begin{cases}
1&\text{if } p_n\ll \sqrt{n},\\
0&\text{if } p_n\gg \sqrt{n},\\
e^{-\alpha^2}&\text{if }\lim_{n\to\infty}\frac{p_n}{\sqrt{n}}=\alpha\in (0,\infty)
\end{cases}\] (see Lemma \ref{lem:intersectionAsymptotics}, a proof is given in the appendix) this $p_n$-spin glass model fulfils the condition of Theorem \ref{thm:convergenceForHypergraphs} if and only if $p_n$ grows slower than $\sqrt{n}$. 

We now turn to the question whether for $p_n\gg \sqrt{n}$, the expected density of states of $H_n^{(p_n-\text{glass})}$ indeed exhibits a different limiting behaviour. As Theorem \ref{thm:pnSPINglass} shows, this is indeed the case and for $p_n$ growing faster than $\sqrt{n}$ the density of states approaches a semicircle distribution. This also shows that the condition about the maximal edge degree in Theorem \ref{thm:convergenceForHypergraphs} is in a certain sense optimal. We start with a combinatorial lemma,
 whose proof is given in the appendix.

\begin{lemma}[Asymptotics of intersections of growing sets]~\label{lem:intersectionAsymptotics}Let $a_n$, $b_n$ and $c_n$ be three 
sequences taking values in $\{1,\dots,n\}$.
\begin{enumerate}[(i)]
\item Given any subsets $A_n\subset\{1,\dots,n\}$ with $a_n$ elements, the proportion of $B_n\subset\{1,\dots,n\}$ with $b_n$ elements that have a non-empty intersection with $A_n$ goes to one if and only if $a_nb_n$ grows faster than $n$. More precisely it holds that
\[\lim_{n\to\infty}\frac{\binom{n}{b_n}-\binom{n-a_n}{b_n}}{\binom{n}{b_n}}=\begin{cases}
1&\text{if } a_nb_n\gg n,\\
0&\text{if } a_nb_n\ll n,\\
1-e^{-\lambda}&\text{if } \lim_{n\to\infty}\frac{a_nb_n}{n}=\lambda\in (0,\infty).
\end{cases}
\]
\item Given any subsets $A_n\subset\{1,\dots,n\}$ with $a_n$ elements, the proportion of $B_n\subset\{1,\dots,n\}$ with $b_n$ elements that share at least $c_n$ elements with $A_n$ goes to 1, i.e.
\[\lim_{n\to\infty}\frac{\abs{\Set{B_n\subset\{1,\dots,n\} | \abs{B_n}=b_n, \abs{A_n\cap B_n}\geq c_n }}}{\binom{n}{b_n} }=1, \] 
 provided $a_n b_n\gg n$ and $c_n\ll \frac{a_nb_n}{n}$.
\end{enumerate}
\end{lemma}

\begin{figure}
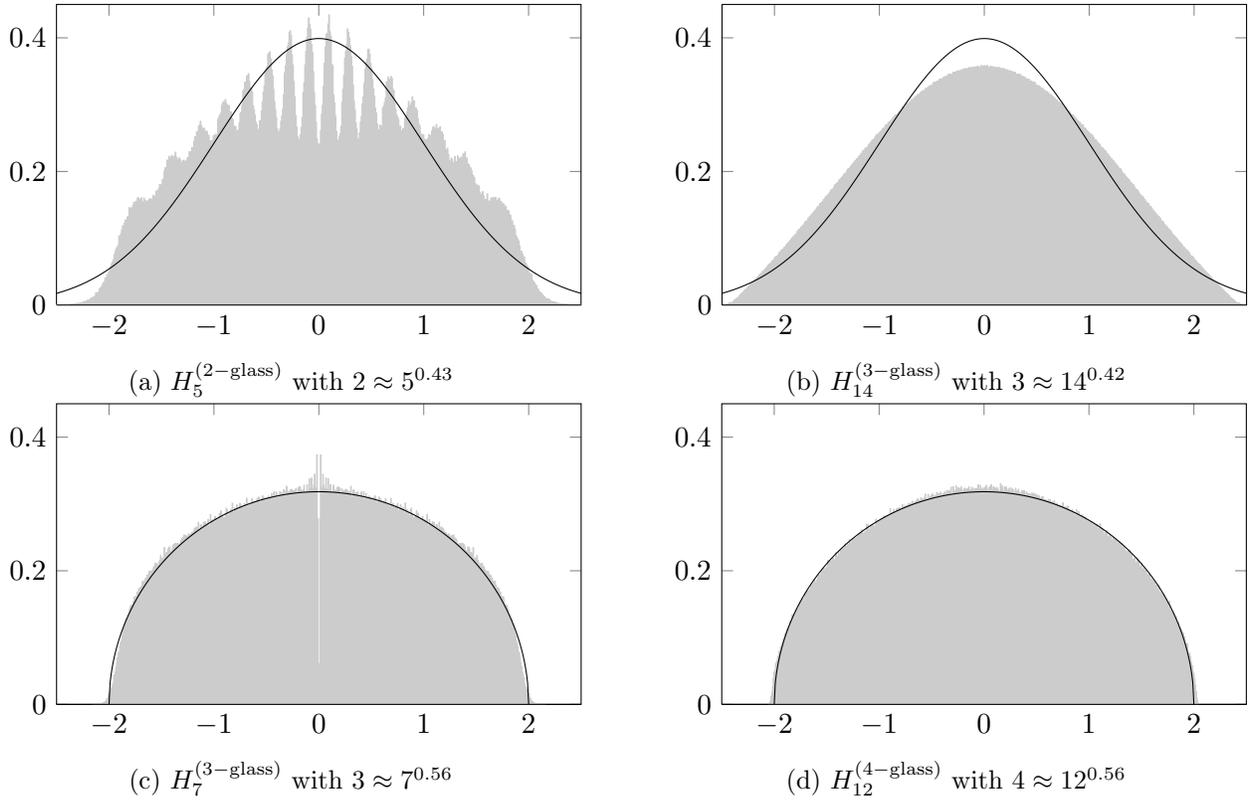

\begin{subfigure}{.49\textwidth}
  \centering
  \setlength\figureheight{4cm} \setlength\figurewidth{.8\linewidth}\input{L5p2GLASS.tikz}
  \caption{$H_{5}^{(2-\text{glass})}$ with $2\approx 5^{0.43}$}
\end{subfigure}
\begin{subfigure}{.49\textwidth}
  \centering
  \setlength\figureheight{4cm} \setlength\figurewidth{.8\linewidth}\input{L14p3GLASS.tikz}
  \caption{$H_{14}^{(3-\text{glass})}$ with $3\approx 14^{0.42}$}
\end{subfigure}\\
\begin{subfigure}{.49\textwidth}
  \centering
  \setlength\figureheight{4cm} \setlength\figurewidth{.8\linewidth}\input{L7p3GLASS.tikz}
  \caption{$H_{7}^{(3-\text{glass})}$ with $3\approx 7^{0.56}$}
\end{subfigure}
\begin{subfigure}{.49\textwidth}
  \centering
  \setlength\figureheight{4cm} \setlength\figurewidth{.8\linewidth}\input{L12p4GLASS.tikz}
  \caption{$H_{12}^{(4-\text{glass})}$ with $4\approx 12^{0.56}$}
\end{subfigure}

\caption{Empirical density of states of some $p$-spin glass Hamiltonian with $p\approx n^{0.42}$ and $p\approx n^{0.56}$}
\label{fig:NumericalDOSpglass}
\end{figure}

\begin{proof}[Proof of Theorem \ref{thm:pnSPINglass}]
As already mentioned the first claim is a immediate consequence of Theorem \ref{thm:convergenceForHypergraphs} and the estimate from Lemma \ref{lem:intersectionAsymptotics}. 

Now assume that $p_n$ grows faster than $\sqrt{n}$. As before, we compute the moments and due to Lemma \ref{lemma:AnkSUM} again know that the odd moments vanish and for even moments we only have to consider those tuples of \begin{align*} J_i\in I_n\defeq\Set{({\bf a},(i_1\dots i_{p_n})) | {\bf a}\in\{1,2,3\}^{p_n}, 1\leq i_1<\dots<i_{p_n}\leq n}\end{align*} which come in pairs of two. Using the already established short hand notation for the $\sigma_{J_i}$ we have, for even $k$, 
 \begin{align*}
m_{k,n}&\approx 3^{-kp_n/2}\binom{n}{p_n}^{-k/2}\sum_{(J_1,\dots,J_k)\in P_2(I_n^k)} 2^{-n}\Tr \sigma_{J_1}\dots\sigma_{J_k}.\end{align*}
 (here $\approx$ means \emph{is equal in the limit $n\to\infty$}).
 We now rephrase condition $(J_1\dots J_k)\in P_2(I_n^k)$. The tuples can be thought of being constructed by first drawing $\frac{k}{2}$ distinct $J_i$ from $I_n$ and then assigning those $\frac{k}{2}$ $J_i$'s to the tuples in a way that each $J_i$ appears twice. By defining the family of (labelled) pair-partitions of the set $\{1,\dots,k\}$ into $\frac{k}{2}$ labelled  subsets with $2$ elements each; \[S_k\defeq\Set{\pi:\{1,\dots,k\}\to\{1,\dots,k/2\} | \abs{\pi^{-1}(\{j\})}=2\text{ for all }1\leq j\leq k/2},\] the sum then reads
\begin{align}\label{eq:mkn} m_{k,n}\approx 3^{-kp_n/2}\binom{n}{p_n}^{-k/2} \sum_{\pi\in S_k} \sum_{\{J_1,\dots,J_{k/2}\}\subset I_n}^\star 2^{-n} \Tr \sigma_{J_{\pi(1)}}\dots \sigma_{J_{\pi(k)}},
\end{align}
 where $\sum^\star$ indicates that the elements $J_1,\dots,J_{k/2}$ are distinct.

At this point it is useful to introduce the notion of non-crossing 
pair-partitions which often appear in random matrix theory. An
element $\pi\in S_k$ shall be called \emph{crossing} if there exists $1\leq a<b<c<d\leq k$ such that $\pi(a)=\pi(c)$ and $\pi(b)=\pi(d)$, otherwise
it is called \emph{non-crossing}; the corresponding subsets of $S_k$ are
denoted by $S_k^{(c)}$ and $S_k^{(nc)}$. These notions emerge in this context since by Lemma \ref{lem:partitionsPauliMatrices} (a proof of which is given in the appendix) for a non-crossing $\pi\in S_k^{(nc)}$ the matrices $\sigma_{J_{\pi(j)}}$ in the trace in~\eqref{eq:mkn} can be reordered such that all appear as squares and therefore the normalised traces are all $1$ independent of the $J_l$'s. 

\begin{lemma}[Product of Pauli matrices ordered in pair-partitions]\label{lem:partitionsPauliMatrices} Let $k$ be even, $\pi\in S_k$ and define \[I_n\defeq \Set{({\bf a},e) | e\in \Gamma_n, {\bf a}\in \{1,2,3\}^{\abs{e}}}\] for some hypergraph $\Gamma_n$.
\begin{enumerate}[(i)]
\item If $\pi$ is non-crossing, then $\frac{1}{2}\Tr \sigma^{(a_{\pi(1)})}\dots \sigma^{(a_{\pi(k)})}=1$ for all $1\leq a_1,\dots,a_{k/2}\leq 3$.
\item If $\pi$ is crossing, there exist $1\leq a_1,\dots a_{k/2}\leq 3$ such that $\frac{1}{2}\Tr \sigma^{(a_{\pi(1)})}\dots \sigma^{(a_{\pi(k)})}\ne 1$.
\item If $\pi$ is non-crossing, then $2^{-n}\Tr \sigma_{J_{\pi(1)}}\dots\sigma_{J_{\pi(k)}}=1$ for all $J_1,\dots,J_{k/2}\in I_n$.
\end{enumerate}
\end{lemma}

For the sum over the non-crossing pair-partitions $S_k^{(nc)}$ we thus find a contribution of
\begin{align*} 3^{-kp_n/2}\binom{n}{p_n}^{-k/2} \sum_{\pi\in S_k^{(nc)}}\sum_{\{J_1,\dots,J_{k/2}\}\subset I_n}^\star 1 
= \frac{\abs{S_k^{(nc)}}\binom{\abs{I_n}}{k/2}}{\binom{n}{p_n}^{k/2} 3^{kp_n/2}}\approx\frac{\abs{S_k^{(nc)}}}{(k/2)!}=\frac{k!}{(k/2)!(k/2+1)!},
\end{align*} 
where it was used that the number of non-crossing pair-partitions into unlabelled subsets are given by the Catalan numbers (see e.g. \cite[Proposition 2.1.11]{anderson2010introduction}). It remains to show that the sum over the crossing  pair-partitions gives no contribution. Since the total number of partitions is finite it suffices to show that 
\begin{equation}\label{eq:sumlim}
\lim_{n\to\infty} \frac{3^{-kp_n/2}}{(k/2)!}\binom{n}{p_n}^{-k/2} \sum_{J_1,\dots,J_{k/2}\in I_n} 2^{-n} \Tr \sigma_{J_{\pi(1)}}\dots \sigma_{J_{\pi(k)}} = 0
\end{equation}
for each crossing $\pi\in S_k$. Notice that this summation is normalised, i.e. the combinatorial prefactor is exactly the number of terms in the sum.

Since $\pi$ is assumed to be crossing there are $1\leq a<b<c<d\leq k$ such that $r\defeq\pi(a)=\pi(c)$ and $s\defeq\pi(b)=\pi(d)$. From Lemma \ref{lem:intersectionAsymptotics} it follows 
that there exists a sequence $q_n \gg 1$ such that the proportion of pairs of subsets of $\{1,\dots,n\}$, with $p_n$ elements each, that share at least $q_n$ elements approaches $1$ as $n\to\infty$. Applied to our normalised sum in \eqref{eq:sumlim}, this means that we can restrict our attention to those terms for which $e_r$ and $e_s$ have at least $q_n$ vertices in common. In this way we arrive at 
\begin{align} &\frac{3^{-kp_n/2}}{(k/2)!}\binom{n}{p_n}^{-k/2} \sum_{({\bf a}_1,e_1),\dots,({\bf a}_{k/2},e_{k/2})\in I_n} 2^{-n} \Tr \sigma_{({\bf a}_{\pi(1)},e_{\pi(1)})}\dots \sigma_{({\bf a}_{\pi(k)},e_{\pi(k)})}\nonumber\\
\approx&\frac{3^{-kp_n/2}}{(k/2)!}\binom{n}{p_n}^{-k/2} \sum_{\substack{e_1,\dots,e_{k/2}\in\Gamma_n\\ \abs{e_r\cap e_s}\geq q_n}} \sum_{{\bf a}_1,\dots,{\bf a}_{k/2}\in \{1,2,3\}^{p_n}}2^{-n}\Tr \sigma_{({\bf a}_{\pi(1)},e_{\pi(1)})}\dots \sigma_{({\bf a}_{\pi(k)},e_{\pi(k)})} \label{eq:crossingPartitionsPn}\end{align} for some $q_n\gg 1$ only depending on $p_n$. 

For a hyperedge $e$, the hyperedge consisting of the first $l$ vertices of $e$ (with respect to the natural ordering) will be denoted by by $e_{1:l}$. 
We introduce the shorthand notation \[ g\defeq(e_r\cap e_s)_{1:q_n},\] (we recall that hyperedges are subsets of the vertex set,
thus set theoretical operations, such as
 $\in,\cup,\cap,\setminus$, are meaningful for them). 
With this notation
 we can factorise the inner sum from eq.~\eqref{eq:crossingPartitionsPn} to get 
\begin{align*}
\sum_{{\bf a}_1,\dots,{\bf a}_{k/2}\in \{1,2,3\}^{p_n}}&2^{-n}\Tr \sigma_{({\bf a}_{\pi(1)},e_{\pi(1)})}\dots \sigma_{({\bf a}_{\pi(k)},e_{\pi(k)})}=\\
&\left(\sum_{{\bf a}_1\in\{1,2,3\}^{\abs{e_1\cap g}}}\dots \sum_{{\bf a}_{k/2}\in\{1,2,3\}^{\abs{e_{k/2}\cap g}}}2^{-n}\Tr \sigma_{({\bf a}_{\pi(1)},e_{\pi(1)}\cap g)}\dots \sigma_{({\bf a}_{\pi(k)},e_{\pi(k)}\cap g)}\right)\\ 
& \times \left(\sum_{{\bf a}_1\in\{1,2,3\}^{\abs{e_1\setminus g}}}\dots \sum_{{\bf a}_{k/2}\in\{1,2,3\}^{\abs{e_{k/2}\setminus g}}}2^{-n}\Tr \sigma_{({\bf a}_{\pi(1)},e_{\pi(1)}\setminus g)}\dots \sigma_{({\bf a}_{\pi(k)},e_{\pi(k)}\setminus g)}\right)
\end{align*} 
where the second factor is bounded by $3^{\sum_{l=1}^{k/2}\abs{e_l\setminus g}}$. We then further factorise the first factor to obtain
\[\prod_{j\in g} \sum_{\substack{a_l\in\{1,2,3\}\text{ if } j\in e_l\cap g\\ a_l=0 \text{ else}}}\frac{1}{2}\Tr \sigma^{(a_{\pi(1)})}\dots\sigma^{(a_{\pi(k)})}.\] 
For any fixed $j\in g$, the number
 \[
m_j\defeq\abs{\Set{1\leq l\leq k/2 | j\in e_l\cap g}}
\] of $l$'s such that $e_l\cap g$ contains the vertex $j$, is always between $2\leq m_j\leq k/2$ since at least $r$ and $s$ satisfy this condition. By ignoring the $a_l=0$ factors and writing $m$ for $m_j$, we see that we can rewrite \[\sum_{\substack{a_l\in\{1,2,3\}\text{ if } j\in e_l\cap g\\ a_l=0 \text{ else}}}\frac{1}{2}\Tr \sigma^{(a_{\pi(1)})}\dots\sigma^{(a_{\pi(k)})}=\sum_{a_1,\dots,a_m=1}^3 \frac{1}{2}\Tr \sigma^{(a_{\tilde{\pi}(1)})}\dots\sigma^{(a_{\tilde{\pi}(2m)})}\] for some crossing $\tilde{\pi}\in S_{2m}$. According to part (ii) of Lemma \ref{lem:partitionsPauliMatrices}, some (but not all) terms in this sum are equal to $-1$ and therefore there exists a (possibly) $\tilde{\pi}$ and $m$-dependent constant $C(\tilde\pi,m)<1$ such that \[\left\lvert \sum_{a_1,\dots,a_m=1}^3 \frac{1}{2}\Tr \sigma^{(a_{\tilde{\pi}(1)})}\dots\sigma^{(a_{\tilde{\pi}(2m)})}\right\lvert\leq C(\tilde\pi,m) \cdot 3^m.\] By setting the \[C\defeq \max_{2\leq m\leq k/2}\max_{\tilde\pi\in S_{2m}^{(c)}} C(\tilde\pi,m)<1\] to be the maximum of those constants, and recalling that $|g|=q_n$, we arrive at
\[\Big\lvert\prod_{j\in g} \sum_{\substack{a_l\in\{1,2,3\}\text{ if } j\in e_l\cap g\\ a_l=0 \text{ else}}}\frac{1}{2}\Tr \sigma^{(a_{\pi(1)})}\dots\sigma^{(a_{\pi(k)})}\Big\lvert\leq \prod_{j\in g} \Big[ C\cdot 3^{\abs{\Set{1\leq l\leq k/2 | j\in e_l\cap g}}}\Big]=C^{q_n}\cdot 3^{\sum_{l=1}^{k/2}\abs{e_l\cap g}}. \] After plugging in our estimates into eq.~\eqref{eq:crossingPartitionsPn} we finally find 
\begin{align}
\Big\lvert \frac{3^{-kp_n/2}}{(k/2)!}&\binom{n}{p_n}^{-k/2} \sum_{({\bf a}_1,e_1),\dots,({\bf a}_{k/2},e_{k/2})\in I_n} 2^{-n} \Tr \sigma_{({\bf a}_{\pi(1)},e_{\pi(1)})}\dots \sigma_{({\bf a}_{\pi(k)},e_{\pi(k)})}\Big\lvert\nonumber\\
&\leq \frac{3^{-kp_n/2}}{(k/2)!}\binom{n}{p_n}^{-k/2} \sum_{\substack{e_1,\dots,e_{k/2}\in\Gamma_n\\ \abs{e_r\cap e_s}\geq q_n}} C^{q_n} 3^{\sum_{l=1}^{k/2}(\abs{e_l\cap g}+\abs{e_l\setminus g})} +\o{1}\nonumber\\
&= \frac{3^{-kp_n/2}}{(k/2)!}\binom{n}{p_n}^{-k/2} \sum_{e_1,\dots,e_{k/2}\in\Gamma_n} C^{q_n} 3^{k p_n/2}+\o{1}=\frac{C^{q_n}}{(k/2)!} +\o{1}=\o{1}\label{eq:crossingPartBound}
\end{align} as $n\to\infty$, proving that the contribution of any crossing partition vanishes.

We have now proved that the $k$-th moment of the limiting distribution is given by $\frac{k!}{(k/2)!(k/2+1)!}$ for even $k$ and $0$ for odd $k$. A direct computation shows that these are the moments of the semicircular distribution with density function \[\rho(x)= \frac{1}{2\pi}\sqrt{4-x^2}\chi_{[-2,2]}(x)\] which furthermore satisfy Carleman's continuity condition.

 We now turn to part (iii) of  Theorem~\ref{thm:pnSPINglass}, i.e. the case where $\lim_{n\to\infty}\frac{p_n}{\sqrt{n}}=\lambda\in(0,\infty)$. By Lemma \ref{lemma:AnkSUM} the odd moments vanish also in this case. For even $k$ an explicit formula for the $k$-th moment can be derived as follows. For a given partition $\pi\in S_k$ we define the \emph{number of crossings} $\kappa(\pi)$ to be the number of subsets $\{r,s\}\subset \{1,\dots,k/2\}$ such that for some $1\leq a<b<c<d\leq k$ we have that $\pi(a)=\pi(c)=r$ and $\pi(b)=\pi(d)=s$. We claim that 
\begin{align} \lim_{n\to\infty}\frac{3^{-kp_n/2}}{(k/2)!}\binom{n}{p_n}^{-k/2} \sum_{({\bf a}_1,e_1),\dots,({\bf a}_{k/2},e_{k/2})\in I_n} 2^{-n} \Tr \sigma_{({\bf a}_{\pi(1)},e_{\pi(1)})}\dots \sigma_{({\bf a}_{\pi(k)},e_{\pi(k)})} = \frac{(e^{-4\lambda/3})^{\kappa(\pi)}}{(k/2)!} \label{eq:kappaPiTrace}\end{align} holds for all partitions $\pi$. If $\{r_1,s_2\},\dots,\{r_{\kappa(\pi)},s_{\kappa(\pi)}\}$ are the crossings of $\pi$, by Lemma \ref{lem:intersectionAsymptotics} the numbers of vertices in the intersections $e_{r_1}\cap e_{s_1},\dots,e_{r_{\kappa(\pi)}}\cap e_{s_{\kappa(\pi)}}$ are approximately independently Poisson-$\lambda$ distributed. It furthermore follows from Lemma \ref{lem:intersectionAsymptotics} that in the limit we can restrict our attention to those edges where the sets $e_{r_1}\cap e_{s_1},\dots,e_{r_{\kappa(\pi)}}\cap e_{s_{\kappa(\pi)}}$ are mutually disjoint. Since the normalised trace of the Hamiltonian acting on a qubit within such a twofold crossing is given by \[3^{-2}\sum_{a,b=1}^3 \frac{1}{2}\Tr\sigma^{(a)}\sigma^{(b)}\sigma^{(a)}\sigma^{(b)} =-\frac{1}{3} \] whereas the normalised trace is $1$ for those qubits not involved in any crossings we find that the lhs. of eq.~\eqref{eq:kappaPiTrace} can be asymptotically rewritten as 
\[\frac{1}{(k/2)!}\sum_{m_1=0}^\infty\dots\sum_{m_{\kappa(\pi)}=0}^\infty \frac{\lambda^{m_1+\dots+m_{\kappa(\pi)}}}{m_1!\dots m_{\kappa(\pi)}!}e^{-\kappa(\pi)\lambda} (-1/3)^{m_1+\dots+m_{\kappa(\pi)}}=\frac{(e^{-4\lambda/3})^{\kappa(\pi)}}{(k/2)!},\] just as claimed. The $k$-th limiting moment, i.e. the normalised trace of $H_n^{(p_n-\text{glass})}$ in the limit $n\to\infty$, is thus given by 
\[m_k(\lambda)\defeq \frac{1}{(k/2)!}\sum_{\pi\in S_k} (e^{-4\lambda/3})^{\kappa(\pi)}=\sum_{\pi\in \tilde{S}_k} (e^{-4\lambda/3})^{\kappa(\pi)},\] where $\tilde{S}_k$ denotes the set of unlabelled partitions. 

These moments uniquely correspond to the distribution given in eq. \eqref{eq:rhoLambda}, as known from the theory of the q-Hermite polynomials, see \cite[eqs. (3.2) and (3.8)]{Ismail1987379}. For the convenience of the reader we collect some further properties of this distribution in Proposition \ref{prop:rhoLambda}.
\end{proof}

\begin{remark}\label{rem:spins}
The proof of the Theorem \ref{thm:pnSPINglass} also works for general spin-$s$ systems (instead of spin-$1/2$) with small changes.  Mainly, part (ii) from Lemma \ref{lem:partitionsPauliMatrices} has to be replaced by a corresponding Lemma for spin-$s$ which can be proved along the lines of the original proof. This replacement (possibly) changes the value of the $C$-constant from eq.~\eqref{eq:crossingPartBound} which is irrelevant for the result since $C<1$ is sufficient for the convergence against zero. In part (iii) the proof also applies to general spin-$s$ systems, except that $e^{-4\lambda/3}$ has to be replaced by $e^{-4s\lambda/(2s+1)}$. Theorems \ref{thm:convergenceForGraphs} and \ref{thm:convergenceForHypergraphs} also carry over to spin-$s$ since for the important bounds only degree properties of the graph and no specifics of the spin-$1/2$ system were used.
\end{remark}

\begin{prop}\label{prop:rhoLambda}
Suppose that $\lim_{n\to\infty}\frac{p_n}{\sqrt{n}}=\lambda\in(0,\infty)$, then 
$m_{k,n}(\lambda)$, the normalised trace of the $k$-th power of $H_n^{(p_n)-\text{glass}}$, in the limit $n\to\infty$ takes the form 
\[
m_k(\lambda)\defeq\lim_{n\to\infty}m_{k,n}(\lambda)=0\] if $k$ is odd and
\begin{align}m_k(\lambda)&\defeq \lim_{n\to\infty} m_{k,n}(\lambda)=
\frac{1}{(1-e^{-4\lambda/3})^{k/2}}\sum_{j=-k/2}^{k/2}(-1)^j e^{-2\lambda\cdot j(j-1)/3}\binom{k}{k/2+j}\nonumber\\ &= \frac{1}{\sqrt{2\pi}} \int_{-\infty}^\infty e^{-x^2/2}\left(\frac{2\sinh^2(i x\sqrt{\lambda/3}+\lambda/3)}{e^{-2\lambda/3}\sinh(-2\lambda/3)}\right)^{k/2}\diff x\label{eq:mkLambda}\end{align}
if $k$ is even. For any fixed even $k$ it furthermore holds that $m_k(\lambda)$ is monotonically decreasing in $\lambda$ and satisfies 
\begin{equation}\label{eq:limrel}
\lim_{\lambda\to 0}m_k(\lambda)=(k-1)!! \quad\text{and}\quad \lim_{\lambda\to \infty}m_k(\lambda)=\frac{k!}{(k/2)!(k/2+1)!} 
\end{equation}
in agreement with the statements of Theorem \ref{thm:pnSPINglass}. The corresponding limiting probability distribution $\mu_\lambda$ has the compactly supported density function given in eq.~\eqref{eq:rhoLambda} which converges pointwise to the semicircular density function when $\lambda\to\infty$ and to the density function of the normal distribution when $\lambda\to0$. Furthermore for any fixed $\lambda$ the density function $\rho_\lambda$ has square root singularities in $\pm2/\sqrt{1-e^{-4\lambda/3}}$.
\end{prop}

\begin{proof}
 Recall from the proof of Theorem \ref{thm:pnSPINglass} that the moments are given by \[m_k(\lambda)=\sum_{\pi\in \tilde{S}_k} (e^{-4\lambda/3})^{\kappa(\pi)}\] for even $k$ and $m_k(\lambda)=0$ for odd $k$. According to an exact formula by Touchard and Riordan and its integral representation (see \cite[eqs. (5) and (7) on page 197]{formalPowerSeries}) the even moments are given by the formulas in eq.~\eqref{eq:mkLambda}. The monotone decrease follows by computing the derivative in $\lambda$ and the claimed limits in \eqref{eq:limrel} are also direct computations. The claimed limiting behaviour of the density function as $\lambda\to 0$ and $\lambda\to\infty$ also follows from the discussion in Section 2 of \cite{Ismail1987379}. The $k=0$ term from eq.~\eqref{eq:rhoLambda} is responsible for the square root singularity near the edges $\pm2/\sqrt{1-e^{-4\lambda/3}}$. 
\end{proof}

\appendix
\section*{Appendix: Proofs of some technical lemmas}
\renewcommand{\thesection}{A} 

\begin{proof}[Proof of Lemma \ref{lemma:pauliMatrices}]~
\begin{enumerate}[(i)]
\item[(i),(ii)] Trivial calculations. \setcounter{enumi}{2}
\item Two Pauli matrices $\sigma^{(a)},\sigma^{(b)}$ anti-commute if $a\not=b$.  We can therefore up to a factor of $\pm 1$ reorder the arguments of $\sigma(a_1,\dots,a_k)$ in such a way that $a_1=\dots=a_{n_1}=0$, $a_{n_1+1}=\dots=a_{n_1+n_2}=1$, $a_{n_1+n_2+1}=\dots=a_{n_1+n_2+n_3}=2$, $a_{n_1+n_2+n_3+1}=\dots=a_{k}=3$ where $n_0+n_1+n_2+n_3=k$ denote the numbers of $0$'s, $1$'s, $2$'s and $3$'s. Using that the square of any Pauli matrix is the identity we then find \[\sigma(a_1,\dots,a_k)=\pm\sigma(\pi_{n_1},2\pi_{n_2},3\pi_{n_3})\] where $\pi_n$ is the parity function i.e. $\pi_n=0$ if $n$ is even and $\pi_n=1$ if $n$ is odd.
\item Using the anti-commutation relation $\sigma^{(a)}\sigma^{(b)}=2\delta_{ab}1_2-\sigma^{(a)}\sigma^{(b)}$ we compute inductively \begin{align*}
\sigma(a_1,a_2,\dots,a_k)&=2\delta_{a_1,a_2}\sigma(a_3,\dots,a_k)-\sigma(a_2,a_1,a_3,\dots,a_k)\\
&=2\delta_{a_1a_2}\sigma(a_3,\dots,a_k)-2\delta_{a_1a_3}\sigma(a_2,a_4,\dots,a_k)+\sigma(a_2,a_3,a_1,a_4,\dots,a_k)\\
&=\dots=2\sum_{j=2}^k (-1)^j\delta_{a_1a_j} \sigma(a_2,\dots,\widehat{a_j},\dots,a_k)-\sigma(a_2,\dots,a_k,a_1) \end{align*} from which the claim follows by the cyclicity of the trace. \qedhere
\end{enumerate}
\end{proof}

\begin{proof}[Proof of Proposition \ref{prop:StarGraph}]
We again start to compute the moments using the short hand notation $\sigma_{J_i}\defeq\sigma_{1}^{(a_i)}\sigma_{j_i}^{(b_i)}$ for $J=(a_i,b_i,j_i)$  and find by Lemma \ref{lemma:AnkSUM} that \begin{align*}
m_{k,n}\approx(9(n-1))^{-k/2}\sum_{\{J_1,\dots,J_{k/2}\}\subset I_n}^\star\sum_{\pi\in S_k} 2^{-n} \Tr\sigma_{J_{\pi(1)}}\dots\sigma_{J_{\pi(k)}}.
\end{align*}
In the limit $n\to\infty$ the part of the sum where two different $J_i$ have the same $j_i$-\emph{coordinate} can be neglected and we find \begin{align*}
m_{k,n}\approx (9(n-1))^{-k/2} \sum_{\{j_1,\dots,j_{k/2}\}\subset\{2,\dots,n\}} \sum_{a_1,\dots,a_{k/2}=1}^3\sum_{b_1,\dots,b_{k/2}=1}^3\sum_{\pi\in S_k} 2^{-1} \Tr\sigma^{(a_{\pi(1)})}\dots\sigma^{(a_{\pi(k)})}
\end{align*} where it was used that in all but the first component the matrices commute (since the $j_i$ are mutually distinct), the square of any Pauli matrix is the identity and that the trace of the tensor product is the product of the traces. After performing the sums over the $j_i$'s and $b_i$'s (and using that $\binom{n-1}{k/2}(n-1)^{-k/2}\approx \frac{1}{(k/2)!}$) we arrive at \[ m_{k,n}\approx m_k\defeq \frac{3^{-k/2}}{(k/2)!}\sum_{a_1,\dots,a_{k/2}=1}^3\sum_{\pi\in S_k}\frac{1}{2} \Tr\sigma^{(a_{\pi(1)})}\dots\sigma^{(a_{\pi(k)})}.\] We claim that \[f(k)\defeq \sum_{a_1,\dots,a_{k/2}=1}^3\sum_{\pi\in S_k} \frac{1}{2}\Tr\sigma^{(a_{\pi(1)})}\dots\sigma^{(a_{\pi(k)})}=\frac{(k+1)!}{2^{k/2}} \] holds for all even $k$. While $k=2$ is trivial, for the induction step we compute using Lemma \ref{lemma:pauliMatrices}(iv) and the notation therein
\begin{align*}
f(k)&=\sum_{a_1,\dots,a_{k/2}=1}^3\sum_{\pi\in S_k} \frac{1}{2}\Tr\sigma^{(a_{\pi(1)})}\dots\sigma^{(a_{\pi(k)})}\\ &=\sum_{j=2}^k (-1)^j\sum_{a_1,\dots,a_{k/2}=1}^3\sum_{\pi\in S_k}\delta_{a_{\pi(1)}a_{\pi(j)}} \sigma(a_{\pi(2)},\dots,\widehat{a_{\pi(j)}},\dots,a_{\pi(k)})
\end{align*} and then split the sum into to parts $f_1(k)$ and $f_2(k)$ where in $f_1(k)$ we only consider those $\pi\in S_k$ for which $\pi(1)=\pi(j)$ and in $f_2(k)$ those $\pi$ for which $\pi(1)\not=\pi(j)$. For $f_1(k)$ we can then compute 
\begin{align*}
f_1(k)&=3\sum_{j=2}^k (-1)^j\sum_{a_1,\dots,\widehat{a_{\pi(1)}},\dots,a_{k/2}=1}^3\sum_{\substack{\pi\in S_k\\ \pi(1)=\pi(j)}} \sigma(a_{\pi(2)},\dots,\widehat{a_{\pi(j)}},\dots,a_{\pi(k)})\\&
= 3\frac{k}{2}\sum_{j=2}^n (-1)^j \sum_{a_1,\dots,a_{(k-2)/2}=1}^3\sum_{\pi\in S_{k-2}} \sigma(a_{\pi(1)},\dots,a_{\pi(k-2)}) = 3\frac{k}{2}\sum_{j=2}^k (-1)^j f(k-2)=3\frac{k}{2}f(k-2)
\end{align*}where in the first step we performed the sum over $a_{\pi(1)}=a_{\pi(j)}$ and in the second step took out a factor of $\frac{k}{2}$ corresponding to the $\frac{k}{2}$ possible values of $\pi(1)=\pi(j)$. Similarly we find that \[f_2(k)=\frac{k(k-2)}{2}f(k-2)\] and by adding the two recursion relations we finally arrive at \[f(k)=f_1(k)+f_2(k)=\frac{k(k+1)}{2}f(k-2)=\frac{k(k+1)}{2}\frac{(k-1)!}{2^{(k-2)/2}}=\frac{(k+1)!}{2^{k/2}}\] proving the claim. Inserting this into the expression we had for $m_k$ then gives $m_k=\frac{(k+1)!}{ 6^{k/2}(k/2)!}$ for even $k$ and $m_k=0$ for odd $k$.

These moments again satisfy Carleman's continuity condition and therefore uniquely correspond to a limiting distribution whose characteristic function $\phi$ is given by \[ \phi(t)=\sum_{k=0}^\infty \frac{(i t)^{2k}}{(2k)!}\frac{(2k+1)!}{6^k k!}=\sum_{k=0}^\infty (-1)^k\frac{(2k+1)t^{2k}}{6^k k!}=\left(1-\frac{t^2}{3}\right)e^{-t^2/6}\] from which by a Fourier transform we find the density \[\rho(x)=\frac{1}{2\pi}\int_{-\infty}^\infty \phi(t)e^{-itx}\diff t=3\sqrt{\frac{3}{2\pi}}x^2e^{-3x^2/2}.\qedhere\]
\end{proof}

\begin{proof}[Proof of Lemma \ref{lem:intersectionAsymptotics}]~
\begin{enumerate}[(i)]
\item We can safely assume that eventually $a_n+b_n\leq n$ since the assertion is trivial otherwise and then compute\[\binom{n-a_n}{b_n}/\binom{n}{b_n}=\frac{(n-a_n)(n-a_n-1)\dots (n-a_n-b_n+1)}{n(n-1)\dots (n-b_n+1)}=\prod_{k=0}^{b_n-1}\left(1-\frac{a_n}{n-k}\right)\] where all factors are non-negative. We continue with the obvious bounds \[\left(1-\frac{a_n}{n}\right)^{b_n}\leq \prod_{k=0}^{b_n-1}\left(1-\frac{a_n}{n-k}\right)\leq \left(1-\frac{a_n}{n-b_n+1}\right)^{b_n} \] and after applying a logarithm arrive at \begin{align*}\exp\left(-\frac{a_nb_n}{n}\right)&\approx\exp\left(b_n\log\left(1-\frac{a_n}{n}\right)\right)\leq \binom{n-a_n}{b_n}/\binom{n}{b_n} \\ &\leq \exp\left(b_n\log\left(1-\frac{a_n}{n-b_n+1}\right)\right)\approx \exp\left(-\frac{a_nb_n}{n-b_n+1}\right)\end{align*} from which the claim follows immediately.
\item For any fixed $k\geq 0$ the proportion of sets $B_n$ of size $b_n$ that share exactly $k$ elements with $A_n$ is given by 
\begin{align*} \binom{a_n}{k} \binom{n-a_n}{b_n-k}/\binom{n}{b_n} \end{align*} i.e. the number of elements in the intersection is hypergeometrically distributed with parameters $(n,a_n,b_n)$ and therefore has a mean of $\frac{a_nb_n}{n}$ and a variance of \[\frac{a_nb_n}{n}\frac{n-b_n}{n}\frac{n-a_n}{n-1} \] which shows that for $c_n$ growing slower than $\frac{a_nb_n}{n}$ the proportion of $B_n$'s that share at least $c_n$ elements with $A_n$ converges to $1$. \qedhere
\end{enumerate}
\end{proof}

\begin{proof}[Proof of Lemma \ref{lem:partitionsPauliMatrices}]~
\begin{enumerate}[(i)]
\item Suppose that $\pi$ is non-crossing. Let $i<j$ be those indices for which $\pi(i)=\pi(j)=1$. 
Since the partition is non-crossing in the tuple $(\pi(1),\dots,\pi(k))$ there are either zero or two $l$ indices between $\pi(i)$ and $\pi(j)$ for all $1<l\leq k/2$. 
Recall that $\sigma^{(a_{1})}$ anti-commutes with $\sigma^{(a_l)}$ if $a_l\not=a_1$ and commutes otherwise. 
Hence we can freely permute $\sigma^{(a_{\pi(j)})}$ to the left next to $\sigma^{(a_{\pi(j)})}$ and then the claim follows inductively since $(\sigma^{(a_1)})^2=1_2$ and we therefore proved the claim assuming the result for $k-2$. 
For $k=2$ the assertion is trivially true.
\item Suppose now that $\pi$ is crossing, i.e. there exist $a<b<c<d$ such that $r\defeq\pi(a)=\pi(c)$ and $s\defeq\pi(b)=\pi(d)$. Then by setting $a_l=1$ for $l\not\in\{r,s\}$ the expression simplifies to \[\frac{1}{2}\Tr \sigma^{(a_{\pi(1)})}\dots \sigma^{(a_{\pi(k)})}=\sigma(a_{\pi(1)},\dots,a_{\pi(k)})=\sigma(\alpha,a_r,\beta,a_s,\gamma,a_r,\delta,a_s,\epsilon)\] for some $\alpha,\beta,\gamma,\delta,\epsilon\in\{0,1\}$. Using the anti-commutation relations we then find that for $a_r,a_s\in\{2,3\}$ it holds that \begin{align*}\sigma(a_{\pi(1)},&\dots,a_{\pi(k)})=(-1)^\gamma \sigma(\alpha,a_r,\beta,a_s,a_r,\gamma,\delta,a_s,\epsilon)=(-1)^{\gamma+1-\delta_{a_r,a_s}}\sigma(\alpha,a_r,\beta,a_r,a_s,\gamma,\delta,a_s,\epsilon)\\&=\dots=(-1)^{2\gamma+\beta+\delta+1-\delta_{a_r,a_s}}\sigma(\alpha,a_r,a_r,\beta,\gamma,a_s,a_s,\delta,\epsilon)=(-1)^{\beta+\delta+1-\delta_{a_r,a_s}}\sigma(\alpha,\beta,\gamma,\delta,\epsilon).\end{align*} i.e. the result changes sign depending on whether $a_r=2$ and $a_s=3$ or $a_r=a_s=2$ and in particular cannot be equal to $1$ for all choices of $a_1,\dots,a_{k/2}$.
\item This is an immediate consequence of applying part (i) to all components of the tensor product separately and using that the trace of a tensor product factorises.\qedhere
\end{enumerate}
\end{proof}

\footnotesize
\bibliographystyle{hplain}
\bibliography{paper_arxiv1}
\end{document}